\documentclass[aps,prl,showpacs, showkeys, preprint,groupedaddress]{revtex4-1}
\usepackage{amssymb}
\usepackage{amsmath}
\usepackage{amsthm}
\usepackage{hyperref}
\usepackage{natbib}

\begin{document}
\newtheorem{theorem}{Theorem}
\newtheorem{corollary}[theorem]{Corollary}
\newtheorem{lemma}[theorem]{Lemma}
\newtheorem{defn}{Definition}
\newcommand{\R}{\ensuremath{\mathbb{R}}}
\newcommand{\Rbar}{\ensuremath{\bar{\mathbb{R}}}}
\newcommand{\Sch}{\ensuremath{\mathcal{S}}}
\newcommand{\D}{\ensuremath{\mathcal{D}}}
\newcommand{\U}{\ensuremath{\mathcal{U}}}
\newcommand{\K}{\ensuremath{\mathcal{K}}}
\newcommand{\N}{\ensuremath{\mathcal{N}}}
\newcommand{\SA}{\ensuremath{\mathcal{SA}}}
\newcommand{\I}{\ensuremath{\mathbf{I}}}
\newcommand{\IRT}{\ensuremath{\mathbf{I}(\R^T)}}
\newcommand{\triple}{\ensuremath{(x,N,I[N])}}

\def\polhk#1{\setbox0=\hbox{#1}{\ooalign{\hidewidth
    \lower1.5ex\hbox{`}\hidewidth\crcr\unhbox0}}} 

%\preprint{}

\title{A global solution to the Schr{\"o}dinger equation: from Henstock to Feynman}


\author{Ekaterina S. Nathanson}
\email[]{enathanson@ggc.edu}
%\thanks{}
\affiliation{School of Science and Technology, \\Georgia Gwinnett College,\\1000 University Center Lane, Lawrenceville, GA 30043}

\author{Palle E. T. J{\o}rgensen}
\email[]{palle-jorgensen@uiowa.edu}
\homepage[]{http://www.math.uiowa.edu/~jorgen/}
%\thanks{}
\affiliation{Department of Mathematics, \\University of Iowa, Iowa City, IA 52242.}


\date{\today}

\begin{abstract}
\indent One of the key elements of Feynman's formulation of non-relativistic quantum mechanics is a so-called Feynman path integral. It plays an important role in the theory, but it appears as a postulate based on intuition rather than a well-defined object. All previous attempts to supply Feynman's theory with rigorous mathematics have not been satisfactory. The difficulty comes from a need to define a measure on the infinite dimensional space of paths and to create an integral that would possess all of the properties requested by Feynman. \\
\indent In the present paper, we consider a new approach to defining the Feynman's path integral, based on the theory developed by P. Muldowney. Muldowney uses the Henstock integration technique, and non-absolute integrability of the Fresnel integrals in order to obtain a representation of the Feynman's path integral as a functional. This approach offers a mathematically rigorous definition supporting Feynman's intuitive derivations. But in his work, Muldowney gives only local in space-time solutions. A physical solution to the non-relativistic Schr{\"o}dinger equation must be global, and it must be given in the form of a unitary one-parameter group in $L^2(\R^n)$. The purpose of this paper is to show that one-dimensional Muldowney's local solutions may be extended to yield a global solution. Moreover, the global extension can be represented by a unitary one-parameter group acting in $L^2(\R)$.
\end{abstract}


%02.30.Cj 	Measure and integration
%02.30.Tb 	Operator theory
%02.30.Sa 	Functional analysis
%03.65.-w Quantum mechanics


\maketitle
\tableofcontents
\section{1. Introduction}
%\subsection{}
%\subsubsection{}
\subsection{1.1 Motivation}
\indent We consider three problems: (I) the definition of the Feynman path integral, (II) an approach to Feynman's question via the Henstock integral, introduced by Muldowney \cite{Mul12}; and (III) a representation of the solution to the non-relativistic Schr{\"o}dinger equation as a path-space average. Our emphasis here is on the word ``average," and we make this precise below.

\indent  In more detail, our first question (I) entails a closer look at Feynman's suggestion for a representation of the solution to the non-relativistic Schr{\"o}dinger equation as a path-space integral \cite{Fey2010, Fey511, Fey512}. By the Schr{\"o}dinger equation, we mean
\begin{eqnarray}
\frac{\partial \psi}{\partial t}&=&\frac{1}{i\hbar}H\psi, \label{eq:Schr_eq}\\
\psi(0,x)&=&\psi_0(x), \nonumber
\end{eqnarray}
where $\psi_0\in L^2(\R^n)$, with $n$-fixed; and the Schr{\"o}dinger operator has the form
\begin{equation}
 H = -\frac{\hbar^2}{2m}\Delta_x+V(x); \label{eq:Schr_op}
\end{equation}
or setting the constant $\hbar=1$, we have
\begin{eqnarray}
\frac{\partial \psi}{\partial t}&=&i\left[\frac{1}{2m}\Delta_x-V(x)\right]\psi, \\
\psi(0,x)&=&\psi_0(x). \nonumber
\end{eqnarray}
By a Feynman integral, we mean a formal representation
\begin{equation}
\psi(x,t)=\textrm{"const"}\int_{\Omega_x}e^{iS(\omega, t)}\psi_0(\omega(t))\D\omega, \label{eq:path_int}
\end{equation}
where $\Omega_x$ is the space of all paths starting at $x$, $\omega:\R\ni t \mapsto \omega(t)\in \R^n$, $\omega(0)=x$.\\
\indent In the more traditional approach in earlier literature to this problem (see e.g., \cite{Nel64}), the effort by many authors has concentrated on answering Feynman's question with variations of what goes by the name ``Feynman integral", and hence a Feynman measure. It is well known, that a mathematical rigorous version of Feynman measure is fraught with difficulties, if not contradictions. While the relevant literature following this idea, and aiming for a Feynman measure, is vast, we mention here only that it is known that there is no sigma-additive positive candidate of Feynman measure; that is ``the Feynman measure'' is not positive measure on path-space. Also, the action term $S(\omega, t)$ in (\ref{eq:path_int}) is given by the time integral $S(\omega,t)=\int_0^t\left(\frac{m}{2}\dot{\omega}^2-V(\omega(s))\right)ds$, and is interpreted only as a formal limit ($n\rightarrow \infty$) of 
\begin{equation}
S_n(x_0,\cdots, x_n,t)=\sum_{j=1}^n\left(\frac{m}{2}\left(\frac{x_j-x_{j-1}}{t/n}\right)^2-V(x_j)\right)\frac{t}{n}. \label{eq:action_term}
\end{equation}
Nelson, in \citep{Nel64}, suggested the following approach to derive the Feyman path integral (\ref{eq:path_int}). Fix $m$ and $V$, and set operators:
\begin{align}
	K^t\psi_0(x)&=\left(\frac{1}{2\pi it}\right)^{1/2} \int_{\mathbf{R}}e^{\frac{i}{2t}|x-y|^2}\psi_0(y)dy,\label{eq:DefOpKt} \\
M^t \psi_0 (x)&=e^{-\frac{i}{2}tV(x)}\psi_0(x). \label{eq:DefOpMt}
	\end{align}
 Then if $H$ in (\ref{eq:Schr_op}) is essentially selfadjoint, Trotter's formula converges in the strong operator topology on $\mathcal{B}(L^2(\R^n))$ to the solution (\ref{eq:path_int})  of the Schr{\"o}dinger equation (\ref{eq:Schr_eq}):
\begin{eqnarray}
	\psi(\cdot,t)&=&U^t\psi_0 \nonumber \\
			&=& \exp\left(it\left(\frac{\Delta}{2m}-V\right)\right)\psi_0\nonumber \\
			& =& \lim_{n  \rightarrow \infty}  \left(K^{t/n}M^{t/n}\right)^n\psi_0. \label{eq:Trotter_solution}
\end{eqnarray}

\subsection{1.2 Our approach}

\indent   Our approach (II) (see theorem \ref{thm:main_result} sec 4.3) serves to reconcile two physical principles. The first one (i) is that time in diffusion is not reversible (the Second Law of Thermodynamics). Mathematically, the corresponding Feynman-Kac formula is a diffusion equation. The second principle (ii) is at the foundation of the meaning of the  Schr{\"o}dinger equation. The Schr{\"o}dinger equation is an equation for the dynamics of wave-functions (representing quantum states). Energy is a conserved quantity, and the Schr{\"o}dinger equation is time reversible; hence represented by a unitary one-parameter group acting in $ L^2(\R^n)$. Our approach is motivated by a diffusion equation; but, to reconcile the two physical principles, (i) and (ii), we must introduce complex ``transition-probabilities." 

\indent Our suggestion here is to give a representation of the solution to the Schr{\"o}dinger equation, not as a path-space integral, but rather as a linear functional on functions defined on path-space. This approach will still yield the space-time solution  to the Schr{\"o}dinger equation, and as an average over path-space (3), but just not with respect to a positive sigma-additive measure. Still our approach has attractive features which are required of a solution to Feynman's question; … perhaps in the form Feynman should have asked it. Our solution is based on the Henstock integral, introduced by Muldowney \cite{Mul12}.   

\indent When the machinery of the Henstock integral is applied to the problem, we do get a solution to the Schr{\"o}dinger equation.  However Muldowney's work offers only local solution, so a solution defined only locally in space-time. Indeed this solution is purely local. But of course a physical solution to the non-relativistic Schr{\"o}dinger equation must be global, and it must be given in the form of a unitary one-parameter  group in $L^2(\R^n)$. Our purpose here is to show that a field of local solutions may be merged, and extended (in fact, uniquely) to yield a unitary one-parameter group $U_t$, where $t$ stands  for time, acting in $L^2(\R^n)$ and producing now a global solution to the Schr{\"o}dinger equation, with energy-conservation, and with time-reversal. 

 \indent Now there are alternative approaches: some based on analytic continuation in time,  and others in mass, thereby making a connection to Wiener measure; and the Feynman-Kac formula for the solution of the inhomogeneous heat equation. Since Wiener measure is indeed well defined as a positive and sigma-additive measure on path-space, this approach has shown some promise, but it has also run up against some difficulties, see e.g., \cite{Nel64}. For example, even in favorable cases, there are analytic continuation from Wiener measure, but valid only for some values of mass parameter in the corresponding Schr{\"o}dinger equation.  Even so, the analytic continuation lacks the kind of direct physical interpretation implicit in Feynman's original question.
\subsection{1.3 The Trotter Formula}
\indent As always, the Hamiltonian operator for the Schr{\"o}dinger equation is a non-commutative sum of a kinetic term, and a potential term. In general such a sum is only formally selfadjoint, but not selfadjoint in the sense of having a unique spectral resolution; the sum may have J. von Neumann defect indices different from (0, 0).  Therefore, as a premise in our theorem, we must assume that this sum is indeed essentially selfadjoint as an operator in $L^2(\R^n)$ (if there are $N$ particles, then $n = 3 N$). Fortunately, there are already powerful results available (see e.g., \cite{Kat95}) with realistic conditions on the potential which imply essential selfadjointness. And hence they apply to all physical potentials. With the tools mentioned above, we then show that the following two solutions agree: a local one and one based on a unitary one-parameter group. The first one, the local solution is based on an application of the Henstock integral, and the other, the global solution, is based on an approximation; in detail, a global time-space solution to the Schr{\"o}dinger equation as a unitary one-parameter group  $U_t$, is achieved with the use of a Trotter-product formula approximation; so an infinite product representation; and with a Lagrangian representation in phase space. The convergence of the Trotter-product formula in turn requires that essentially selfadjointness holds for the Hamiltonian operator under consideration.  

\indent The study of the Schr{\"o}dinger operator and their realization as generators of one-parameter groups of unitary operators has been studied over decades. While these approaches vary, they are all different from ours. The questions addressed differ from ones of these investigations. To give a sample, we mention the following papers, and the works cited there \cite{MR0450732, MR0432063, MR0438178, MR0383469, MR0210416, MR0211713, MR0175537, MR0167152, MR0161190, MR0450732, MR2535462, MR1703898, MR1331982, MR0290680, MR0054135,MR1817625,   MR1817619, MR1771173}.

\section {2 Preliminaries}

\subsection{2.1 Henstock Integral in Finite-Dimensional Space}

\indent In this section, we will define the Henstock integral over a domain of a finite-dimensional space, $\R^n$. The set-up and most of the notation is taken from \cite{Mul12}.
\begin{defn} A cell or cylindrical interval $I$ of \R$^n$ is defined to be 
	$$I=I(N)=I_1\times\cdots\times I_n,$$
where  $N=\{1,\dots,n\}$ is the list of the dimensions; $I_{j}$, $1\leq j \leq n$, are one-dimensional intervals, that can be one of the following intervals: $(-\infty,a]$, $(a,b]$, $(b,\infty)$, $(-\infty, \infty)$. The collection of all cells in $\R^n$ is denoted by $\mathbf{I}(\R^n)=\{I(N) \}$. The union of a finite number of cells is called a figure and denoted be $E$. The collection of all figures in $\R^n$ is denoted by $\mathbf{E}(\R^n)$.
\end{defn}
To consider the Riemann sums for an integral over $\R^n$ we associate with each cell $I(N)$ a sample point or a tag: $x=(x_1,\dots,x_n)\in\R^n$. The sample point together with the cell form an associated pair $(x, I(N))$. We choose a sample point to be one of the vertices of the corresponding cell $I(N)$, that is for each $1\le j\le n$ we have\\
	\hspace*{30pt} i) $x_j=-\infty$ if $I_j=(-\infty,a]$,\\
	\hspace*{30pt} ii) $x_j=a$ or $x_j=b$ if $I_j=(a,b]$,\\
	\hspace*{30pt} iii) $x_j=\infty$ if $I_j=(b,\infty),$\\
	\hspace*{30pt} vi) $x_j=-\infty$ or $x_j=\infty$ if $I_j=(-\infty, \infty)$.\\
Since the tag points are allowed to take values of infinity, we need to introduce a real line with ``points at infinity" adjoint. Let's denote it by $\Rbar=\R\cup\{-\infty,\infty\}$.
\begin{defn}
	A gauge in \Rbar$^n$ is a positive function $\delta:\Rbar^n\rightarrow (0,\infty)$. An associated pair $(x,I(N))$ is said to be $\delta$-fine with respect to the given gauge $\delta(x)$, if for all $1\leq j\leq n$: $(x_j,I_j)$ is $\delta$-fine in \R, i.e.\\
	\hspace*{30pt}i) if $I_j=(a_j,b_j]$ then $|I_j|=b_j-a_j<\delta(x)$,\\
	\hspace*{30pt}ii) if $I_j=(-\infty, a_j]$ then $a_j<-\frac{1}{\delta(-\infty)}$,\\
	\hspace*{30pt}iii) if $I_j=(b_j, \infty)$ then $b_j>\frac{1}{\delta(\infty)}$.
\end{defn}
For the finite intervals $I_j=(a_j,b_j]$, the condition  $|I_j|<\delta(x)$ means that in the limiting process the length of the interval $I_j$ gets smaller as $\delta(x)$ becomes smaller. This condition is similar to the condition for choosing partitions in the definition of the proper Riemann integral on a bounded domain. The main difference is in the value for the bound. In the case of the Riemann integral, the bound is constant for each of the subintervals of a partition or the maximum of lengths of the subintervals is bounded by a constant. In the case of the Henstock integral, we replace the constant bound with a value of a gauge function $\delta(x)$. The gauge function $\delta(x)$ can take different values depending on a particular sample point $x$ and any particular properties of an integrand. This special feature of a gauge function makes some non-integrable functions in the Riemann or Lebesgue sense integrable in the sense of Henstock.
\\
\indent In the case of unbounded intervals, $I_j=(-\infty, a_j]$ and $I_j=(b_j, \infty)$, the corresponding conditions, $a_j<-\frac{1}{\delta(-\infty)}$ and $b_j>\frac{1}{\delta(\infty)}$,  imply that as $\delta(-\infty)$ and $\delta(\infty)$ become smaller, the absolute value of an upper boundary for $(-\infty, a_j]$ and the value of a lower boundary for $(b_j, \infty)$ become larger . Thus the unbounded intervals $I_j=(-\infty, a_j]$ and $I_j=(b_j, \infty)$ ``shrink" in size.
\begin{defn}
A partition in $\R^n$ is a finite collection $\mathcal{P}$ of disjoint cells $I(N)$ whose union is  $\R^n$. A division $\D$ in $\R^n$ is a finite collection of associated point-cell pairs $(x, I(N))$ whose cells $I(N)$ form a partition of $\R^n$. Given a gauge $\delta$ in $\Rbar^n$, a division $\D$ is called $\delta$-fine if each of the pairs $(x, I(N))\in\D$ is $\delta$-fine.
\end{defn}
An integrand in $\R^n$ is a function of associated points and cells. It is more intuitive to think about an integrand in the form of a product: $f(x)h(I(N))$. In this case a function $h(I(N))$ plays a role of a measure on the cells. But an integrand can be also given in a more general form $h(x, I(N))$. 
\begin{defn}
A function $h(x,I(N))$ is integrable on \R$^n$ in the Henstock sense with integral
	$$\alpha = \int_{\R^n}h(x,I(N)),$$
	if, given $\varepsilon>0$ there exists a gauge  $\delta$ in $\Rbar^n$ so that, for each $\delta$-fine division $D_\delta$ of $\R^n$, the corresponding Riemann sums satisfy
	$$\left|(D_\delta)\sum h(x,I(N))-\alpha\right|<\varepsilon$$
	or more explicitly 
	$$\left|\sum \{h(x,I(N)): (x,I(N)) \in D_\delta\}-\alpha\right|<\varepsilon.$$
\end{defn}
\indent We defined the Henstock integral of a function with domain of integration being the whole space $\R^n$, but similarly the definition of  the integral can be formulated for a figure $E$, a subset of $\R^n$.

\indent The definition of the Henstock integral relies on the fact that for every positive function $\delta(x)$, a $\delta$ - fine division $\D$ of the domain of integration can be found. This result is due to the Belgium mathematician Pierre Cousin (1895) and is known in the literature as Cousin's lemma or Cousin's theorem. The proof in a one-dimensional case can be found, for example, in \cite{Yee2000} or \cite{Mul12}. In the finite-dimensional case of $\R^n$ with $n>1$, the result can be obtained by applying the lemma in each dimension separately.

\subsection{2.2 Henstock Integral in Infinite-Dimensional Space}

In the previous section, we gave the definition of the Henstock integral in the finite-dimensional case. The gauge function $\delta(x)$ and the $\delta$-fine property of the point-cell pairs are essential for selecting appropriate partition in the definition of the Henstock integral in $\R^n$. For the case of infinite dimensions we will need to introduce more advanced kind of gauge and adjust gauge-fine property accordingly. In this section, we introduce necessary changes and then define the Henstock integral in the infinite-dimensional case.
\\
\indent For the purpose of applying the Henstock theory to the Feynman path integral, the domain of integration that we will consider, is $\R^T$. $\R^T$ is a product of $T$ copies of \R. By $T$ we denote a real interval, a subset of $(0,+\infty)$. Usually, we assume $T=(0 ,t]$  or $T=(0,+\infty)$. It would help our intuition if we would think of $\R^T$ as of the space of all real-valued functions $x(t)$ defined on $T$:
$$\R^T:=\{x~|~ x:T \mapsto \R\}.$$
Note, $\R^T$ includes both continuous and discontinuous functions, so it is not equal to the space of all paths, that are represented by continuous functions.\\
\indent Let us denote by $\N=\N(T)=\{N=\{t_1,\dots,t_n\}\subseteq T, n\in \mathbb{N}\}$ a collection of all finite subsets of points in $T$. Suppose all the sets $N=\{t_1,\dots,t_n\}$ are ordered in the increasing manner: $t_1<t_2<\dots<t_n$. 
\begin{defn} A cell or cylindrical interval $I$ of \R$^T$ is defined to be 
	\begin{align}
	I=I[N]&=I_1\times\cdots\times I_n\times \R^{T\backslash N} \nonumber \\
		  &= I_1\times\cdots\times I_n\times \prod\left\{\R^{\{t\}}~|~ t\in T\backslash N\right\}\nonumber
	\end{align}
where  $N=\{t_1,\dots,t_n\}\in \N(T)$, $I_j=I_{t_j}$ for $1\le j \le n$. Similarly to the finite-dimensional case, $I_j=I_{t_j}$ are the one-dimensional intervals of the form: $(-\infty,a]$, $(a,b]$, $(b,\infty)$ or $(-\infty, \infty)$. 
\end{defn}
\indent Note that to distinguish a finite-dimensional cell in \R$^N$ from an infinite-dimensional one in \R$^T$, we use parenthesis in the notation in the finite-dimensional case:
	$$I(N)=I_1\times\cdots\times I_n$$
	and brackets in the infinite-dimensional one. Thus a cell in \R$^T$ can be written in the following form:
	$$I[N]=I(N)\times \R^{T\backslash N}.$$
The collection of all cells in $\R^T$ is denoted by $\mathbf{I}(\R^T)=\{I[N] ~|~ N\in \N(T)\}$. A figure in $\R^T$ is the union of a finite number of cells, denoted by $E$. The collection of all figures in $\R^T$ is $\mathbf{E}(\R^T)$.
	\\
	\indent The form of functions that we will consider later, makes it appropriate to take a triple instead of a pair as a basic element of a partition. We consider triples with included set of restricted dimensions of the form $(x,N,I[N])$. The triple $(x,N,I[N])$ is called associated if $I[N]=I_1\times \cdots \times I_n \times \R^{T\backslash N}$, $N=\{t_1,\dots,t_n\}\in \N(T)$ and $x(N)=(x(t_1),\dots, x(t_n))$ is a vertex of $I_1\times \cdots \times I_n$, i.e. the corresponding finite-dimensional pair  $((x(t_1),\dots, x(t_n)), I(N))$ is an associated pair in $\R^N$. Note that if $(x, N, I[N])$ are associated and $t\in T\backslash N$, then $x(t)\in \R^{\{t\}}$ can be an arbitrary real number. 
\begin{defn}
	A gauge in \R$^T$ $\gamma$ is a pair of mappings $(L, \delta)$:
	\begin{align}
		L: \Rbar^T &\rightarrow \N(T) \nonumber\\
			x &\mapsto L(x) \nonumber
	\end{align}
and
	\begin{align}
		\delta: \Rbar^T\times \N(T) &\rightarrow (0,+\infty) \nonumber\\
			(x, N) &\mapsto \delta(x,N).	\nonumber
	\end{align}
 With $N=\{t_1,\dots, t_n\}\in \N(T)$, an associated triple $(x,N,I[N])$ is $\gamma=(L, \delta)$-fine or fine with respect to the gauge $\gamma=(L, \delta)$ if
	\begin{enumerate}
	\item $L(x)\subseteq N$.
	\item For all $t_j\in N$: $(x(t_j),I_j)$ is $\delta$-fine, i.e.\\
	i) if $I_j=(a_j,b_j]$ then $|I_j|=b_j-a_j<\delta(x,N)$,\\
	ii) if $I_j=(-\infty, a_j]$ then $a_j<-\frac{1}{\delta(-\infty,~N)}$,\\
	iii) if $I_j=(b_j, +\infty)$ then $b_j>\frac{1}{\delta(+\infty,~N)}$.
	\end{enumerate}
\end{defn}
The first condition above means that the set $L(x)$ is included in the dimensions where the cell $I[N]$ is restricted. The second one requires that the finite dimensional point-cell pair $((x(t_1),\dots, x(t_n)), I(N))$  is fine with respect to gauge $\delta$.

\indent Similarly to the finite-dimensional case, a division in $\R^T$ is a finite collection $D$ of point-cell pairs $(x,I[N])$ such that the corresponding triples \triple~are associated and the cells $I[N]$ are disjoint with union $\R^T$.
\begin{defn}
	If a gauge $\gamma=(L,\delta)$ is given, then a division $D$ is $\gamma$-fine, if each \triple$\in D$ is $\gamma$-fine.
\end{defn}
\begin{defn}
	A function $h$ of associated triples \triple~is integrable in \R$^T$ in the Henstock sense with integral
	$$\alpha = \int_{\R^T}h\triple,$$
	if, given $\varepsilon>0$ there exists a gauge  $\gamma$ so that, for each $\gamma$-fine division $D_\gamma$ of $\R^T$, the corresponding Riemann sums satisfy
	$$\left|(D_\gamma)\sum h\triple-\alpha\right|<\varepsilon$$
	or more explicitly 
	$$\left|\sum \left\{h\triple: \triple \in D_\gamma\right\}-\alpha\right|<\varepsilon.$$
\end{defn}
\indent Note that in the case of $\R^T$, the rule $\gamma$ for selecting associated triples in Riemann sums has two conditions. We still use a positive function $\delta(x)$ to bound the lengths of the restricted edges of each cell $I[N]$, but we also have an extra condition for the sets of restricted dimensions $N$. In the finite-dimensional integral over $\R^n$, the set of restricted dimensions was always fixed. Now it becomes a variable. The condition requires, that the set of restricted dimensions $N$ for each cell $I[N]$ in the partition, includes some minimal set of dimensions, given by a value of the gauge function $L(x)$ at the corresponding sample point. Thus if we make $\delta(x)$ successively smaller and $L(x)$ successively larger the cells in the corresponding $\gamma$ - fine associated triples, will be ``shrinking." This is the idea in the limiting process of the Henstock integral in the infinite-dimensional space. 
\\
\indent Besides gauges, the specific feature of the definition of the Hensock integral is the implicit fact, that for a given gauge $\gamma$, there exists a $\gamma$ - fine division $D_\gamma$ of $\R^T$, as in the finite-dimensional case. However, this time the proof is more technical and requires some extra tools from the theory of the Henstock integral.
\begin{theorem}
	Given any gauge $\gamma=(L,\delta)$ in $\R^T$, there exists a $\gamma$ - fine division $D_\gamma$ of $\R^T$.
\end{theorem}
\begin{proof}
 	See Theorem 4 in chapter 4 in \cite{Mul12}.
\end{proof}

\subsection{2.3 Elements of Henstock - Probability Theory}
 \begin{defn}
	A real- or complex-valued function $F$ defined on the figures $\mathbf{E}(\R^T)$ of $\R^T$ is an additive cell function if $F$ is finitely additive on disjoint figures. An additive cell function $F$ is a probability distribution function if $F(\R^T)=1$.
\end{defn}
In classical probability theory a probability distribution function is also assumed to be non-negative. But for the case of the Feynman path integral, we drop this requirement and allow the distribution functions to be negative or complex-valued. This assumption makes our approach different from the classical one, when probability distributions are non-negative and correspond to the well-defined measures. Since we already know, that the classical approach does not guarantee satisfactory results for Feynman's theory, our method based on including non-positive distribution functions, is justified. The motivation for our approach comes from a simple experiment with a beam of photons, and our argument lies at the foundation of the wave-particle duality; i.e., the idea that elementary particles (such as photons-light quanta) behave simultaneously as particles and as waves.

\indent The experiment goes as follows: a laser beam is directed towards a particle detector-device. Using sensors we can measure a number of light particles that reach the detector with a certain probability. Since the particles also possess the wave properties, depending on their phase, the interaction between photons can result in cancellation, or intensification, of the beam. Thus a probability distribution function that measures the number of particles must then necessarily take on negative as well as positive values. This explains why, in quantum mechanics, it is reasonable to have a probability distribution function which attains complex values.
\begin{defn}
	If $F$ is a probability distribution function on $\R^T$ and if there exists a function $f(x(N))$ such that for each $I\in \I(\R^T)$ 
		$$F(I)=\int_I f(x(N))|I[N]|$$
	then $f$ is a density function for $F$.
Here by $|I[N]|=|I_1\times \dots \times I_n\times \R^{T\backslash N}|$ we understand the finite product
$$|I[N]|=|I_1|\cdot|I_2|\dots|I_n|.$$
\end{defn}
\vspace*{12pt}

\section{3 Local Solution to the Schr{\"o}dinger Equation}
In his work Muldowney considers a general version of the Schr{\"o}dinger equation in one-dimension of the form
$$
\frac{\partial \psi(x,t)}{\partial t}+\frac{1}{4c}\frac{\partial^2\psi(x,t)}{\partial x^2}+cV(x,t)\psi(x,t)=0.
$$
Making the parameter $c$ to be $-\frac{i}{2}$ and the potential to depend only on space variable $V(x,t)=V(x)$, we arrive to the Schr{\"o}dinger equation in commonly used form with mass parameter taken to be 1 and $\hbar=1$:
\begin{equation}
\frac{\partial \psi(x,t)}{\partial t}=i\left[\frac{1}{2}\frac{\partial^2}{\partial x^2}-\frac{1}{2}V(x)\right]\psi(x,t). \label{eq:Schrod_gen}
\end{equation} 
Below is the statement of the theorem as it appears in the work by Muldowney with the only change that the parameter $c$ is taken to be $-{i\over 2}$. 
\begin{theorem} (Muldowney, 2012)\label{thm:solShrodingerEq}
	Let $x\in\R$ and $t>0$ be given. Suppose the function $V(y)$ is continuous at $y=x$. If
	$$\psi(x,t)=E^*_{xt}[\U(x_{T^-})]$$
exists in a neighborhood of $(x,t)$ then $\psi(x,t)$ satisfies the  Schr{\"o}dinger equation
	\begin{equation}
	\frac{\partial \psi(x,t)}{\partial t}=i\left[\frac{1}{2}\frac{\partial^2}{\partial x^2}-\frac{1}{2}V(x)\right]\psi(x,t). \label{eq:ShrodingerEq}
	\end{equation}
\end{theorem}
Here $E^*_{xt}[\U(x_{T^-})]$ is the marginal expectation given by
\begin{equation}
E^*_{xt}[\U(x_{T^-})]  = \int_{\R^{T^-}}\U(x_{T^-})G(I[N^-]). \label{eq:MargExp}
\end{equation}
In this work, $x_T$ represents a possible path of a particle; and is a real-valued function on time interval $x_T:[0,t]\rightarrow \mathbb{R}$. By definition of the marginal expectation the paths that are considered in the integral 
(\ref{eq:MargExp}) are conditioned  to have the right end-points fixed. The notation $x_{T^-}$ is used to distinguish the path with a fixed right end from a path with a free right end. Thus the domain of integration for $E^*_{xt}[\U(x_{T^-})]$  is given by
$$ \R^{T^-}=\{x(t)~|~x:[0,t]\rightarrow \R, x(0)=0, x(t)=x \} $$
and the integrands have the following definitions
\begin{eqnarray}
     \U(x_{T^-})&=& \left\{\begin{array}{rcl}  \exp(-\frac{i}{2}\int_{T^-}V(x(t))dt), &x_{T} \textrm{ is continuous}\\
                  0, & \textrm{otherwise}. \nonumber
	\end{array} \right.\\ \nonumber
		G(J[N^-])&:=& \int_{J(N^-)}g(x(N^-))|J(N^-)| \nonumber\\
							&=&   \left. \int_{J(N^-)}g(x(N^-))\right\vert_{x(0)=0, x(t)=x} dx_1\dots dx_{n-1} \nonumber\\
							&= & \prod_{j=1}^{n-1}  \left. \left(\sqrt{\frac{-i}{2\pi(t_j-t_{j-1})}}\right) \int_{I_j} e^{\frac{i}{2}\frac{(x_j-x_{j-1})^2}{(t_j-t_{j-1})}}\right\vert_{x(0)=0, x(t)=x} dx_j. \label{eq:fnc_G}
\end{eqnarray}
\indent The function $\exp\left(\frac{i}{2}\frac{(x_j-x_{j-1})^2}{(t_j-t_{j-1})}\right)$ under the integral sign is not Lebesgue integrable.  However, we can use a closed contour in the complex plane to compute the improper Riemann integral for this function (see also  \cite{MR2362774,MR2038334, MR2005680, MR1703428}). In one-dimension, the improper Riemann integral is equivalent to the Henstock integral \cite{Mul12}. Thus the Henstock integral of the Fresnel integral with the simple expression in the exponent can be evaluated as follows
$$\int_\R e^{\frac{i}{2}x^2}dx_1=\sqrt{\frac{2\pi}{-i}}.$$
The method with a closed contour in the complex plane for one-dimensional case lies in the core of the proof for the result stated in the theorem below. Using the appropriate change of variables, the Riemann sum for the infinite dimensional integral can be computed as a product of one-dimensional integrals. We refer interested reader to \cite{Mul12} for details.
\begin{theorem}
The Fresnel infinite-dimensional integrand is integrable in $\R^T$ and
$$\int_{\R^T} G (I[N]) =1.$$
\end{theorem}
The function $G$ (\ref{eq:fnc_G}) is additive and has integral 1 over $\R^T$, thus it is the complex ``transition probability."  It represents the probability of a path to go through the window $I_1$ at time $t_1$, $I_2$ at time $t_2$ and so on. The windows $I_1,\dots, I_{n-1}$ are specified by the cylindrical interval $J[N^-]=J(N^-)\times \R^{T\backslash N^{-}}=I_1\times \dots\times I_{n-1}\times \R^{T\backslash N^{-}}$. Since all the path in the integral (\ref{eq:MargExp}) have the right end fixed, there is no integration over the last instance $t$. We use the notation $N^-=\{t_1,\dots,t_{n-1}\}$ for the list of points in time excluding the last one. 

\indent Note, that $G$ is not probability function in the classical sense. Based on the rules of quantum mechanics, we state that it is reasonable to introduce ``probability distribution functions" which take on complex values.  Also we would like to remark, that the solution to the Schr{\"o}dinger equation given by theorem \ref{thm:solShrodingerEq} possesses a  property of a Green function:
$$\lim_{t\rightarrow 0} E^*_{xt}[\U(x_{T^-})]=\delta(x).$$
Since the paths  $x_T=x:[0,t]\rightarrow \mathbb{R} $ have fixed both left and right ends, it follows from the definition of the marginal expectation $E^*_{xt}[\U(x_{T^-})]$ that if $t$ tends to zero, then  $E^*_{xt}[\U(x_{T^-})]$ is zero unless $x=0$. \\
\indent Also, the solution $ E^*_{xt}[\U(x_{T^-})]$ is purely local, it is defined only in a space-time neighborhood of $(x,t)$. In the next section we extend it uniquely to a global solution and give it representation in an operator form. First, let us introduce some notation: denote by $K$ a compact set. Let $\D_K$ be a subset of $C_c^\infty(\R)$-functions that have their support in $K$:
	$\D_K=\{\phi\in C_c^\infty(\R)| \textrm{supp}~ \phi\subseteq K\}$; let $\D=\cup_K \D_K$.
\begin{corollary} \label{cor:corollary1}
Let $K\subset\R$ be a compact set, $0<t<\tau$ be fixed, $f(x)\in\D_K$, $V(y)$ be continuous at $y=x\in K$, then there exist a neighborhood of $(x,t) \in K\times(0,\tau]$ where the function
\begin{equation}
\psi^{loc}(x,t) = \int_{\R}E^*_{xt}[\U(x_{T^-})]f(y)dy \label{eq:henstock_sol}
\end{equation}
is a solution to the following IVP 
\begin{align}
\frac{\partial \psi(x,t)}{\partial t}&=i\left[\frac{1}{2}\frac{\partial^2}{\partial x^2}-\frac{1}{2}V(x)\right]\psi(x,t), \label{eq:IVP_loc1}\\
	\psi(x,0) &= f(x). \quad  \label{eq:IVP_loc2}
\end{align}
\end{corollary}
\begin{proof}
First, we show that $\psi^{loc}(x,t)$ satisfies the initial condition (\ref{eq:IVP_loc2}):
\begin{align}
\left.\psi^{loc}(x,t)\right\vert_{t=0} &=\left. \left(\int_{\R}E^*_{xt}[\U(y_{T^-})]\right\vert_{t=0}f(y)dy\right)   \nonumber \\ 
			&  =\lim_{t\rightarrow 0}\left( \int_{\R} E^*_{xt}[\U(y_{T^-})] f(y)dy\right).   \nonumber 
\end{align}
As $t$ tends to zero  the condition on the integrand becomes $x(t)\rightarrow x(0)$, thus the marginal expectation $E^*_{xt}$ yields one if and only if x=y and otherwise it yields zero, that means
\begin{align}
\left.\psi^{loc}(x,t)\right\vert_{t=0} &=\int_{\R}\left. E^*_{xt}[\U(y_{T^-})] \right\vert_{t=0} f(y)dy \nonumber \\ 
			& =\int_{\R}\delta(x-y) f(y)dy  \nonumber \\ 
			& = f(x).  \nonumber
\end{align}
Thus the initial condition is satisfied. Next we prove that $\psi^{loc}(x,t)$ solves the equation (\ref{eq:IVP_loc1}). We use the fact that $\psi^{loc}(x,t)$ is a locally integrable function. We have
\begin{align}
i\left(\frac{1}{2}\frac{\partial^2}{\partial x^2}-\frac{1}{2}V\right)\psi^{loc}(x,t) &=i\left(\frac{1}{2}\frac{\partial^2}{\partial x^2}-\frac{1}{2}V\right)\int_{\R}E^*_{xt}[\U(x_{T^-})]f(y)dy \nonumber \\
  &= \int_{\R}i\left(\frac{1}{2}\frac{\partial^2}{\partial x^2}-\frac{1}{2}V\right)E^*_{xt}[\U(x_{T^-})]f(y)dy \nonumber \\
&= \int_{\R} \frac{\partial}{\partial t}E^*_{xt}[\U(x_{T^-})]f(y)dy \quad \textrm{by theorem  \ref{thm:solShrodingerEq} }\nonumber\\
&= \frac{\partial}{\partial t}\int_{\R}E^*_{xt}[\U(x_{T^-})]f(y)dy \nonumber \\
& = \frac{\partial}{\partial t}\psi^{loc}(x,t).  \nonumber
\end{align}
\end{proof}
\section{4 Global solution to the Schr{\"o}dinger Equation}
\subsection{4.1 Class of potentials $\SA(\R)$}

\indent Our goal is to compare the solution for the Schr{\"o}dinger equation via Henstock (\ref{eq:henstock_sol}) to the solution via Trotter's theorem  in the form of a one-parameter unitary group (\ref{eq:Trotter_solution}). Since the convergence of the limit in Trotter's formula  (\ref{eq:Trotter_solution}) requires the Hamiltonian to be essentially selfadjoint, we limit ourselves to considering a special class of functions serving as potentials. 
\begin{defn}
  Let $V:\R\rightarrow\R$ be a real measurable function. We say that $V$ belongs to the class $\SA(\R)$ if and only if the operator $\left(-\Delta+V\right)$ is essentially selfadjoint on $L^2(\R)$.
\end{defn}

   \indent The question of giving precise mathematical conditions for essential selfadjointness of physical Schr{\"o}dinger operators $\left(-\Delta+V\right)$  has a long history going back to the two pioneering papers of T. Kato \cite*{Kat511, Kat512} in 1951. In a formulation given by J. von Neumann, the question is about the distinction between what in the math physics literature is called formally selfadjoint (alias, Hermitian, or symmetric), as opposed to selfadjoint. Typically the initial dense domain where the given Schr{\"o}dinger operator $\left(-\Delta+V\right)$ is naturally defined needs to be completed in a graph closure, yielding thus the closure of  $\left(-\Delta+V\right)$, as a closed Hermitian operator. If the closure is selfadjoint, we say that $\left(-\Delta+V\right)$ is essentially selfadjoint. In general settings, it may not be. But it is the latter selfadjointness condition which allows us to write a spectral resolution for $\left(-\Delta+V\right)$, and also to create effective path space representations for the associated Schr{{\"o}dinger equation. \\
    \indent   By J. von Neumann's theory, we know that $\left(-\Delta+V\right)$ is essentially selfadjoint if and only if it has indices (0, 0). While the von Neumann indices have a physical interpretation, to decide essential selfadjointness is always a subtle question, both from the standpoint of mathematics and physics. We are interested in effective path space representations for the Schr{\"o}dinger equation. But since we rely on first having a resolution to the essential selfadjointness question, we include here the following citations \cite*{Kat511,Kat512,Kat58,Rad78,Kat52,Kal98,Don97}. Each of these papers, starting with Kato 1951, offers useful insight into the question. For concrete examples we give the following conditions on potential function $V$ that guarantee that it falls into the class $\SA(\R) $:
\begin{enumerate}
\item $V(x)>0$ and $\lim_{|x|\rightarrow \infty}V(x)=\infty$ \cite{Ree75}.
\item $-\Delta+V$ has a dense set of analytic vectors \cite{Nel59}.
\item  For a system of $N+1$ particles in 3-dimensional space denote the position vector of $i$th particle by $r_i=(x_i,y_i,z_i)$, then $V$ can be expressed as 
	$$V(r_1, \cdots, r_N) = V'(r_1,\cdots,t_N)+\sum_{i=1}^NV_{0,i}(r_i)+\sum_{i<j}^{1,N}V_{ij}(r_i-r_j),$$
where $V'$ is bounded in the whole configuration space, and $V_{i,j} (0\leq i< j\leq s)$, defined in 3-dimensional space, are locally square-integrable and bounded at infinity (see \cite{Kat511}). In other words, there exist two constants $C$ and $R$ such that
\begin{eqnarray}
	|V'(r_1,\cdots, r_N)|&\leq&C, \nonumber\\
	\int_{r\leq R} |V_{ij}(x,y,z)|^2dxdydz &\leq& C^2, \quad (0\leq i< j\leq N)\nonumber \\
 	|V_{ij}(x,y,z)|&\leq&C \quad (r>R, r=(x^2+y^2+z^2)^{1/2},  0\leq i< j\leq N).\nonumber 
\end{eqnarray}
\end{enumerate}

\indent An example of the potential for which $\left(-\Delta+V\right)$ is not essentially selfadjoint on $L^2(\R)$ is $V(x)=-x^4$ (see \cite{MR0383469}). The reason is that the classical motion reaches $\pm\infty$ in finite time: fix an energy $E>0$ then $\left(\frac{dx}{dt}\right)^2-x^4=E$ yields 
$$t_\infty=\int_0^\infty \frac{dx}{\sqrt{E+x^4}}<\infty.$$

\subsection{4.2 Operator $\mathcal{K}_t(x,f(\cdot))$}
\indent Let $\phi\in\D$, $\{O_i\}$ be a finite open cover of supp~$\phi$ and $\beta_i$ be a partition of unity subordinate to $\{O_i\}$, i.e. the following conditions are satisfied: (1)~$\beta_i\geq 0$, (2)~$\beta_i \in C^\infty_c(\R)$, (3)~supp~$\beta_i\subset O_i$, (4)~$\sum \beta_i(x)=1$ for all $x\in \textrm{supp}~\beta_i$. Then we can write $\phi=\sum \beta_i\phi$ and define a linear map $\mathcal{K}_t(x,\cdot):\D_K\rightarrow\D'$ as follows $$\mathcal{K}_t(x,f(\cdot)):= \int_{\R}E^*_{xt}[\U(x_{T^-})]f(y)dy$$ with an action on a test function $\phi(x)\in\D$:
\begin{align}
<\mathcal{K}_t(x,f(\cdot)),\phi(x)> &=\sum<\mathcal{K}_t(x,f(\cdot)),\beta_i\phi(x)>.  \label{eq:_ext_onD_K}
\end{align}
In the context of the corollary \ref{cor:corollary1}, $\mathcal{K}_t(x,f(\cdot))$ exists only locally in $x$, so to define its  action as a distribution on $\D$, we use the partition of unity. Representing  $\phi=\sum \beta_i\phi$ allows us to pair $\mathcal{K}_t(x,f(\cdot))$ with $\beta_i\phi(x)$ that are non-zero in a sufficiently small neighborhood that can be chosen according to the corollary \ref{cor:corollary1}. Thus using linearity we define $<\mathcal{K}_t(x,f(\cdot)),\phi(x)>$ as a sum of pairing of $\mathcal{K}_t(x,f(\cdot))$ with $\beta_i\phi(x)$, so the expression (\ref{eq:_ext_onD_K}) makes sense.\\
\indent Next corollary extends the domain of $\mathcal{K}_t(x,f(\cdot))$ to $L^2(\R)$.
\begin{corollary} \label{cor:corollary2}
 Under the conditions in the statement of the corollary \ref{cor:corollary1} with $f(x)\in L^2(\R)$ the function
\begin{equation}
\psi^{loc}(x,t) = \mathcal{K}_t(x,f(\cdot)) \label{eq:henstock_sol_L2}
\end{equation}
is a local solution to the IVP  (\ref{eq:IVP_loc1})-(\ref{eq:IVP_loc2}).

\end{corollary}
\begin{proof}
	Let $f(x)\in L^2(\R)$, then we can write $f_n(x)=f(x)\Big|_{[-n,n]}$ and
$$f(x)=\lim_{n\rightarrow \infty} f_n(x).$$
	For all $n\in \mathbb{N}$, $[-n,n]$ is compact, thus $\mathcal{K}_t(x,f_n(\cdot))$ is well defined and by corollary \ref{cor:corollary1} is a local solution to (\ref{eq:IVP_loc1})-(\ref{eq:IVP_loc2}) with the initial condition $\psi(x,0) =f_n(x)$. By interchanging limits with the integral sign in the proof of the corollary \ref{cor:corollary1}, we can show that the solution to  (\ref{eq:IVP_loc1})-(\ref{eq:IVP_loc2}) for $f(x)\in L^2(\R)$ is obtained as follows:
$$\psi^{loc}(x,t)=\mathcal{K}_t(x,f(\cdot))=\lim_{n\rightarrow \infty} \mathcal{K}_t(x,f_n(\cdot)).$$
\end{proof}

\subsection{4.3 Main result}

\begin{theorem} \label{thm:main_result} If ~$V(x)$ is a continuous function in $\SA(\R)$ then each of the local solutions to the Schr{\"o}dinger equation $\psi^{loc}(x,t)$ given by  corollary \ref{cor:corollary2} has a unique extension $\psi^{glob}(x,t)$ for all $t\in\R$ and $x\in\R$. The extension $\psi^{glob}(x,t)$ solves the following IVP with $f\in L^2(\R)$:
\begin{align}
\frac{\partial \psi^{glob}(x,t)}{\partial t}&=i\left[\frac{1}{2}\frac{\partial^2}{\partial x^2}-\frac{1}{2}V(x)\right]\psi^{glob}(x,t), \quad t\in\R, x\in\R, \label{eq:IVP_glob1}\\
	\psi^{glob}(x,0) &= f(x), \quad x\in\R. \label{eq:IVP_glob2}
\end{align}
Moreover, $\psi^{glob}(x,t)$ is represented by a unitary one-parameter strongly continuous group in $L^2(\R)$:
\begin{equation}
\psi^{glob}(x,t)=\mathcal{K}_t(x,f(\cdot)).
\end{equation}
\end{theorem}
\begin{proof}
\indent For $V\in\SA(\R)$ the operator $-\Delta+V$ is essentially selfadjoint, thus by Stone's theorem  \cite{Nel69} the solution to (\ref{eq:IVP_glob1})-(\ref{eq:IVP_glob2}) is given as a strongly continuous group of unitary operators $\psi(x,t)=(U^tf)(x)$ acting in $L^2(\R)$. The operator $U^t$ is related to the operators $K^t$ and $M^t$ (\ref{eq:DefOpKt})-(\ref{eq:DefOpMt}) via Trotter's limit in the strong operator topology \cite{Nel64}:
\begin{eqnarray}
	\psi(x,t)&=&	(U^tf)(x)  \nonumber\\
		    &=&\lim_{n\rightarrow\infty}\left(K_m^{t/n}M_V^{t/n}\right)^n f(x)  \nonumber\\
		 &=&\lim_{n\rightarrow\infty} \left(\frac{n}{2\pi it}\right)^{\frac{n}2}\underbrace{\int_{\R}\cdots\int_{\R}}_{n \textrm{ times}} \prod_{j=1}^n e^{\left[i\frac{n}{2t}\left|x_j-x_{j-1}\right|^2-i\frac t nV(x_j)\right]}f(x_n)dx_1\dots dx_n \nonumber \\
		 &=& \lim_{n\rightarrow\infty} \left(\frac{n}{2\pi it}\right)^{\frac{n}2}\underbrace{\int_{\R}\cdots\int_{\R}}_{n \textrm{ times}} e^{i\sum_{j=1}^n\left[\frac{m}{2}\frac{|x_j-x_{j-1}|^2}{(t/n)^2}-V(x_j)\right]\frac t n}f(x_n)dx_1\dots dx_n.  \nonumber 
	\end{eqnarray}
\indent By corollary \ref{cor:corollary2} there exist a local solution to the IVP  (\ref{eq:IVP_glob1})-(\ref{eq:IVP_glob2}): $\psi^{loc}(x,t) = K_t(x,f(\cdot))$, for small $t$. We would like to compare the local solution $\psi^{loc}(x,t)$ and the global Trotter solution $(U^tf)(x)$. 
\begin{lemma} 
\label{lemma:local_match}
Under the conditions in the statement of the theorem for any $f\otimes\phi\in L^2(\R)\otimes \D$ there exist $\varepsilon>0$ such that for all $|t|<\varepsilon$
  $$<\left(U^tf\right)(x),\phi(x)> = < \K_t(x,f(\cdot)),\phi(x)>.$$
\end{lemma}
\begin{proof}
 \indent By the Schwartz kernel theorem (e.g. see theorem 5.2.1 in \cite{Hor83}) for the linear map $\mathcal{K}_t(x,\cdot):L^2(\R)\rightarrow\D'$ defined in (\ref{eq:_ext_onD_K}), there is one and only one distribution $\widetilde{K}$ on $L^2(\R)\otimes\D$ such that 
\[<\mathcal{K}_t(x,f(\cdot)),\phi(x)> =\widetilde{K}_t(f\otimes\phi)
\]
for all  functions $f(x)\in L^2(\R)$, $\phi(x)\in\D$. Then the following is well-defined for $|s|\le|t|<\varepsilon$
 \begin{align}
\widetilde{K}_s((U^{t-s}f)\otimes \phi)&=<\K_s(x,U^{t-s}f),\phi> \nonumber \\
&=\begin{cases} <\K_0(x,U^tf(\cdot)),\phi(x)>=<\left(U^tf\right)(x),\phi(x)>,&\mbox{if } s = 0 \\ 
< \K_t(x,f(\cdot)),\phi(x)>, & \mbox{if } s =t. \end{cases}  \nonumber
\end{align}
Differentiate this expression in variable $s$, using the corollary \ref{cor:corollary2}:
\begin{align}
 \frac{d}{ds}<\K_s(x,U^{t-s}f),\phi> &= \lim_{h\rightarrow 0}\frac{<\K_{s+h}(x,U^{t-s-h}f),\phi>-<\K_s(x,U^{t-s}f),\phi>}{h}\nonumber \\ 
&=\lim_{h\rightarrow 0}\frac{<\K_{s+h}(x,U^{t-s-h}f),\phi>-<\K_s(x,U^{t-s-h}f),\phi>}{h}\nonumber \\ 
&+ \lim_{h\rightarrow 0}\frac{<\K_s(x,U^{t-s-h}f),\phi>-<\K_s(x,U^{t-s}f),\phi>}{h}\nonumber \\ 
&=\frac{i}{2}\left[  <\left(\Delta-V\right)\mathcal{K}_s(x,\left(U^{t-s}f\right)),\phi> \right. \nonumber \\
&+\left.\left(  <- \left(\Delta-V\right) \mathcal{K}_s(x,\left(U^{t-s}f\right)),\phi>\right)\right] \nonumber\\
&=0. \nonumber
\end{align}
Since the derivative is constant, we arrive to the equality for all $f(x)\in L^2(\R)$, $\phi(x)\in\D$ and $|t|<\varepsilon$:
$$<\left(U^tf\right)(x),\phi(x)> = < \K_t(x,f(\cdot)),\phi(x)>.$$
\end{proof}
Since the operator $U^t(\cdot)$ coincides with $K_t(x,\cdot)$ for small $|t|<\varepsilon$, where $\varepsilon$ is from lemma \ref{lemma:local_match}, we get $K_t(x, f(\cdot))\in L^2_{\textrm{loc}}(\R)$. Also $U^t(\cdot)$ is unitary and  possesses the semigroup property, thus $K_t(x,\cdot)$ is also unitary and locally has the semigroup property 
\begin{equation}
K_t(x,K_s(x,\cdot))=K_{t+s}(x,\cdot) \label{eq:semigroup_prop}
\end{equation}
for $|t|$, $|s|$, $|t+s|<\varepsilon$. Now we can extend the local solution $\K_t(x,f(\cdot))$ in time. For any $t\in\R$, there exist an integer $m$ such that $|t/ m|<\varepsilon$. Then by the semigroup property we have
$$\mathcal{K}_t(x,f(\cdot)) = (\mathcal{K}_{t/m}(x,f(\cdot)))^m=\mathcal{K}_{t/m}(x,\cdots\mathcal{K}_{t/m}(x,f(\cdot))\cdots).$$
Local semigroup (\ref{eq:semigroup_prop}) and unitary properties allow us to extend $\K_t(x,f(\cdot))$ globally as a one-parameter unitary group. To carry out the steps for going from local semigroup to a global unitary one-parameter group see \cite{MR874059, MR863534, MR911991, MR759101, MR727083}. $\K_t(x,f(\cdot))$ is also strongly continuous w.r.t.  $L^2$ since $U^t$ possesses this property.\\
\indent Hence, in conclusion, we get the global solution to the Schr{\"o}dinger equation (\ref{eq:IVP_glob1})-(\ref{eq:IVP_glob2}) realized as a Muldowney path-space average via its local "complex integrals." This solution is a strongly continuous group of unitary operators acting in $L^2(\R)$:
 $$\psi^{glob}(x,t):=\mathcal{K}_t(x,f(\cdot)).$$
\end{proof}

\section{5. Summary}

\indent As a summary, we give the outline of the technical steps in the proof of the main result. Because of the assumed essential selfadjointness, we automatically have the existence of the unitary one-parameter group $U^t$ from the spectral theorem, so a solution to the Schr{\"o}dinger equation as a strongly continuous group of unitary operators $U^t$ acting in $L^2(\R)$. And this unitary one-parameter group $U^t$ is therefore realized as a Trotter approximation. The latter allows us in turn to also realize it locally as Muldowney path-space average $\mathcal{K}_t$.\\ 
  \indent However, the drawback is that while the latter  $\mathcal{K}_t$   arises from the average operation based on Feynman's intuitive idea, and on the rigorous Muldowney path-space calculus, it does not yield operators in $L^2(\R)$. This is because it is based on path-space averages with complex densities. So  $\mathcal{K}_t$  only yields a local solution to the Schr{\"o}dinger equation as a system of Schwartz-distribution kernels, and so it does not yield operators in the Hilbert space $L^2(\R)$. The $L^2$ realization exists only via the unitary one parameter group $U^t$, so the two must be connected in a precise way. This requires the machinery from \cite{MR874059, MR863534, MR911991, MR759101, MR727083} which in turn allows a precise linking from local to global. \\
\indent So our piecing together the local path-space formulas to a global solution to the Schr{\"o}dinger equation is accomplished in the following sequence of steps, and with us making use of the local-to-global theory from \cite{MR874059, MR863534, MR911991, MR759101, MR727083} as follows. We first show that: (i)  $U^t$  and $\mathcal{K}_t$  agree locally on test functions.  (ii) Using the one-parameter group-property for $U^t$ , we conclude that  $K_t$  satisfies a local semigroup property.  (iii) Since $U^t$  is unitary in $L^2$  for all t in $\R$, we conclude that $\mathcal{K}_t$  extends canonically from initially being only a local semigroup, to a global realization as a unitary one-parameter group. (iv) Strong continuity w.r.t.   $L^2$ is automatic since $U^t$ is already strongly continuous.\\
\indent In conclusion, we would like to remark that the results of the main theorem \ref{thm:main_result} hold in more general settings for any potential $V$ for which the Trotter's limit is convergent (\ref{eq:Trotter_solution}). The restriction $V$ to the class $\mathcal{SA}(\R)$ guarantees the convergence of the Trotter's limit, but it's not an equivalent condition. \\

\begin{acknowledgments}
For helpful suggestions and fruitful discussions, the authors are extremely grateful to Professor Wayne Polyzou (Dept of Physics, Univ of Iowa), and to the members of the Mathematical Physics seminar at the Univ of Iowa.
\end{acknowledgments}

 
\bibstyle{plain}
\bibliographystyle{apsrev4-1}
\bibliography{my_biblio}
\end{document}